\documentclass[envcountsame]{llncs}
\usepackage[utf8]{inputenc}
\usepackage{microtype}
\usepackage{amsmath,amssymb}
\usepackage{multirow}
\usepackage{graphicx}
\usepackage{array}
\usepackage{float}
\usepackage{hyperref}
\usepackage{wrapfig}
\usepackage{subcaption}
\usepackage{nicefrac}
\usepackage{algorithm,algpseudocode}

\title{Dense Steiner problems:\\ Approximation algorithms and inapproximability}
\author{
  Marek Karpinski\inst{1} \and
  Mateusz Lewandowski\inst{2} \and
  Syed Mohammad Meesum\inst{2} \and
  Matthias Mnich\inst{3}
}
\institute{
  Universit{\"a}t Bonn, Institut f{\"u}r Informatik, Bonn, Germany
  \and
  University of Wroc{\l}aw, Institute of Computer Science, Wroc{\l}aw, Poland
  \and
  TU Hamburg, Institute for Algorithms and Complexity, Hamburg, Germany
 }

\spnewtheorem{fact}[theorem]{Fact}{\bfseries}{\itshape}
\spnewtheorem{observation}[theorem]{Observation}{\bfseries}{\itshape}

\begin{document}

\maketitle

\begin{abstract}
  The {\sc Steiner Tree} problem is a classical problem in combinatorial optimization: the goal is to connect a set $T$ of terminals in a graph $G$ by a tree of minimum size.
  Karpinski and Zelikovsky (1996) studied the $\delta$-dense version of {\sc Steiner Tree}, where each terminal has at least $\delta |V(G)\setminus T|$ neighbours outside $T$, for a fixed $\delta > 0$.
  They gave a PTAS for this problem.
  
  We study a generalization of pairwise $\delta$-dense {\sc Steiner Forest}, which asks for a mini\-mum-size forest in $G$ in which the nodes in each terminal set $T_1,\dots,T_k$ are connected, and every terminal in $T_i$ has at least $\delta |T_j|$ neighbours in $T_j$, and at least $\delta|S|$ nodes in $S = V(G)\setminus (T_1\cup\dots\cup T_k)$, for each $i, j$ in $\{1,\dots, k\}$ with $i\neq j$.
  Our first result is a polynomial-time approximation scheme for all $\delta > 1/2$.
  Then, we show a $(\frac{13}{12}+\varepsilon)$-approximation algorithm for $\delta = 1/2$ and any $\varepsilon > 0$.
  
  We also consider the $\delta$-dense {\sc Group Steiner Tree} problem as defined by Hauptmann and show that the problem is $\mathsf{APX}$-hard.
\end{abstract}

\section{Introduction}
\label{sec:introduction}
The {\sc Steiner Tree} problem in graphs is a classical problem in combinatorial optimization.
It takes as input a graph $G$ together with a distinguished subset $T\subseteq V(G)$ of nodes called \emph{terminals}, and seeks a minimum-size tree which connects all terminals in $T$.
The currently best-known approximation factor achievable in polynomial time is $\ln(4) + \varepsilon$ for any $\varepsilon > 0$\footnote{That algorithm approximates the minimum-weight solution in graphs with non-negative edge weights.}, and is due to Byrka et al.~\cite{ByrkaGRS2013}.
Of considerable interest is also the generalization of {\sc Steiner Tree} to multiple terminal sets $T_1,\dots,T_t$, where the goal is to find a minimum-size forest such that all terminals of each set $T_i$ belong to the same connected component.
For this problem, which is known as {\sc Steiner Forest}, a 2-approximate solution can be found in polynomial time~\cite{AgrawalKR1995,WilliamsonGMV1995}.
The nodes in $S:=V(G)\setminus(T_1,\dots,T_k)$ are called \emph{Steiner nodes}.

Karpinski and Zelikovsky~\cite{KarpinskiZ1997} considered the $\delta$-dense version of {\sc Steiner Tree}, where every terminal has at least $\delta |V(G)\setminus T|$ neighbours outside~$T$, for some fixed $\delta > 0$.
They gave a PTAS for this problem.
Subsequently, Hauptmann~\cite{Hauptmann2013} studied dense versions of problems related to {\sc Steiner Tree}, such as {\sc Prize-Collecting Steiner Tree}, {\sc Group Steiner Tree}, {\sc $k$-Steiner Tree}, and claimed polynomial-time approximation schemes (PTAS) for these problems.
He also obtained a $1+O\left(\frac{\sum_i^k \log(|T_i|)}{\sum_i^k |T_i|}\right)$-approximation for a certain version of {\sc Steiner Forest} where each terminal in $T_i$ is adjacent to at least $\delta|V(G)\setminus T_i|$ nodes in $V(G)\setminus T_i$.
All his results use the star contractions technique in which, after a number of greedy contraction steps which reduce the number of terminals, the remaining problem is solved optimally.

It is open if {\sc Steiner Tree} is $\mathsf{NP}$-hard for $\delta$-dense inputs.
Likewise, the complexity of {\sc Steiner Forest} for pairwise $\delta$-dense instances is open.

\medskip
\noindent
\textbf{Our contributions.}
We shed light on the exact and approximate solvability of pairwise $\delta$-dense {\sc Steiner Forest}.
\begin{theorem}
\label{thm:more-than-1-2-dense}
  Let $\delta > \frac{1}{2}$. {\sc Steiner Forest} admits a polynomial-time approximation scheme on pairwise $\delta$-dense instances.
  On pairwise $\delta$-dense instances without Steiner nodes, an optimal solution can be found in polynomial time.
\end{theorem}
We give the proof in \autoref{sec:more-than-1-2}.
\begin{theorem}
\label{thm:1-2-pairwise-dense}
  {\sc Steiner Forest} admits a polynomial-time $(\frac{13}{12}+\varepsilon)$-approximation for any $\varepsilon > 0$, for all pairwise $\frac{1}{2}$-dense instances.
\end{theorem}
We give the proof in \autoref{sec:1-2-dense}. 
Our results show a connection between dense instances of {\sc Steiner Forest} and {\sc Set Packing} problem.
It is an interesting open question whether {\sc Steiner Forest} admits a PTAS on pairwise $\delta$-dense instances for all values of $\delta$.

Finally, in \autoref{sec:group-st} we consider {\sc Group Steiner Tree} on dense instances.
The {\sc Group Steiner Tree} problem takes as input a graph $G$ and terminal sets $T_1,\dots,T_t$, and the goal is to find a smallest subtree of~$G$ containing at least one node from each terminal set group $T_i$ for $i = 1,\dots,t$.
An instance is \emph{$\delta$-dense} if each terminal $v$ is adjacent to at least $\delta |V(G)\setminus (T_1\cup\dots\cup T_t)|$ nodes.
We give a reduction from {\sc Set Cover} to {\sc Group Steiner Tree} on dense instances.
The reduction implies that $\delta$-dense {\sc Group Steiner Tree} is $\mathsf{APX}$-hard.
The reader may contrast this with the PTAS that Hauptmann~\cite{Hauptmann2013} claimed for this problem.
Recall that if $\mathsf{P}\not=\mathsf{NP}$, no $\mathsf{APX}$-hard problem admits a PTAS.

\section{Preparations}
\label{sec:preliminaries}
Consider an instance of {\sc Steiner Forest} consisting of a graph $G$ and pairwise disjoint terminal sets $T_1,\dots,T_t$.
We first justify that we can work with instances without Steiner nodes.
For {\sc Steiner Forest} with general (non-negative) edge weights, this assumption can be made without loss of generality, as one could replace every Steiner vertex with a pair of terminals connected by a zero-weight edge.
For {\sc Steiner Forest} with unit edge weights, we first claim that such instances are hard without density condition.
This is captured by following theorem, whose proof we give in Appendix~\ref{sec:proof-unit-weight-thm}.

\begin{theorem}
\label{thm:unit-weight-apx-hard}
  Unit-weight {\sc Steiner Forest} without Steiner nodes is $\mathsf{APX}$-hard.
\end{theorem}

We consider $\delta$-dense instances of {\sc Steiner Forest}.
Formally, an instance $(G;T_1,\dots,T_k)$ of {\sc Steiner Forest} is \emph{pairwise $\delta$-dense} if each terminal $t\in T_i$ is adjacent to at least $\delta|T_j|$ nodes in $T_j$ and at least $\delta|S|$ nodes in $S = V(G)\setminus (T_1\cup\dots\cup T_k)$, for every $i = 1,\dots,k$ and $j\not=i$.

\begin{lemma}
\label{thm:reduction-no-steiner}
  For all $\delta >0$, any pairwise $\delta$-dense instance $(G;T_1,\dots,T_k)$ of {\sc Steiner Forest} admits an optimal solution $F^\star$ that either (i) has at most one connected component, or (ii) does not use Steiner nodes.
\end{lemma}
\begin{proof}
  Consider an optimal forest $F^\star$ for $(G;T_1,\dots,T_k)$ which consists of trees $H^\star_1,\dots, H^\star_\ell$.
  If $\ell = 1$, the claim holds.
  If $\ell\geq 2$, consider a tree $H^\star \in F^\star$ which contains at least one Steiner node.
  Without loss of generality, \mbox{$H^\star=H^\star_1$}.
  Let $S^\star_1$ be the set of Steiner nodes in~$H^\star_1$.
  Consider any tree $H^\star_2\not=H^\star_1$.
  By pairwise density, each terminal $t \in H^\star_1$ has an edge in $G$ to some terminal in $H^\star_2$.
  Thus, the graph induced by the nodes $V(H^\star_1) \cup V(H^\star_2) \setminus S^\star_1$ is connected.
  So we can remove $H^\star_1,H^\star_2$ from~$F^\star$, and replace them by a single tree spanning $V(H^\star_1) \cup V(H^\star_2) \setminus S^\star_1$.
  As $|S^\star_1|\geq 1$, this operation does not increase the value of the solution.
  Repeatedly applying this operation leads to an optimal solution which either does not use any Steiner nodes or has one connected component.
\qed
\end{proof}

The following observation is straightforward.
\begin{fact}
\label{fact:max-cc}
   Let $G$ be a graph and let $F$ be a forest in $G$.
   Then $|E(F)| = |V(G)| - \mathsf{comp}(F)$, where $\mathsf{comp}(F)$ denotes the number of trees in $F$.
\end{fact}
This means that we can view our problem equivalently as that of maximizing the total number of connected components.

Consider any tree $H$ that is part of feasible solution to a {\sc Steiner Forest} instance.
We define $\mathsf{rank}(H)$ as the number of terminal sets that $H$ spans.
Further, call a terminal set $T_i$ \emph{trivial} if it induces a connected subgraph of~$G$.
For each terminal $t\in T_i$, let $N_G(t)$ denote the neighbors of $t$ in $G$.

\section{Approximation Scheme for Pairwise $(>\frac{1}{2})$-Dense Steiner Forest}
\label{sec:more-than-1-2}
In this section we give the algorithm that yields the proof of \autoref{thm:more-than-1-2-dense}.
Let $\delta > \frac{1}{2}$ and consider a pairwise $\delta$-dense instance $(G;T_1,\dots,T_k)$ of {\sc Steiner Forest}.
We search for an optimal solution $F^\star$ which satisfies one of the properties of \autoref{thm:reduction-no-steiner}.

To approximate an optimal solution $F^\star$ which consists of a single connected component, we create a single terminal set $T = \bigcup_{i=1}^{\ell}T_i$.
Note that $(G;T)$ is a $\delta$-dense instance of {\sc Steiner Tree}, as each node $t\in T$ is adjacent to at least $\delta|S|$ nodes of $S = V(G)\setminus T$.
Thus, we fix some $\varepsilon > 0$ and run the Karpinski-Zelikovsky polynomial-time approximation scheme for {\sc Steiner Tree} on $\delta$-dense instances on input $(G;T)$ and $\varepsilon$, to obtain a $(1+\varepsilon)$-approximate solution $F$ to $F^\star$.

From now on, we focus on finding an optimal solution to $(G;T_1,\dots,T_k)$ which does not use any Steiner nodes.
This allows us to assume that $V(G)=\bigcup_{i=1}^k T_i$.
Recall that $\delta > \frac{1}{2}$; so we can make the following useful observation:
\begin{observation}
\label{observation:connected-pair}
  The subgraph $H_{ij}$ of the input graph $G$ which is induced by any two terminal sets $T_i,T_j$ is connected.
\end{observation}
\begin{proof}
  Consider any two nodes $u,v \in T_i$.
  As $\delta > \frac{1}{2}$, $u$ has more than $\frac{|T_j|}{2}$ neighbours in $T_j$.
  The same can be said about $v$, and thus $u$ and $v$ share a common neighbour in $T_j$.
  By symmetry, any two vertices in $T_j$ are also connected in $G[T_i \cup T_j]$.
  Clearly, there is at least one edge between terminals in $T_i$ and $T_j$, hence the observation follows.
\qed
\end{proof}

An immediate corollary of the above observation is that for any two terminal sets $T_i,T_j$, there exists in $G$ a tree of rank 2 spanning $T_i\cup T_j$.
This suggests the greedy \autoref{greedy-more-1-2}.
\begin{algorithm}[h!]
  \caption{Greedy algorithm for pairwise $(>\frac{1}{2})$-dense Steiner forest\label{greedy-more-1-2}}
    \begin{algorithmic}[1]
      \Require A graph $G$ and terminal sets $T_1,\dots,T_k$ such that for each $t\in T_i$ and each $j \in [k]$, $|N_{T_j}(t)| > |T_j|/2$.
      \Ensure A Steiner forest $F$ for $(G;T_1,\dots,T_k)$ of minimum size.
      \item If $G = G[T_1]$ and $G$ is disconnected, then output ``No solution" and exit.
      \item Set $F=\emptyset$.
      \item Add to $F$ spanning trees of every trivial set $T$.
      \item Arbitrarily pair remaining terminal sets, and add corresponding rank-2 spanning trees to $F$.
      \item If there is a single set terminal $T$ left, add it to any tree in $F$ (choose a tree of larger rank).
    \end{algorithmic}
\end{algorithm}

We claim that \autoref{greedy-more-1-2} returns optimal solutions.
Consider an optimal forest $F^\star$ which consists of trees $H^\star_1,\dots, H^\star_k$.
Observe that $\mathsf{rank}(H^\star_i) \leq 3$, for otherwise we could replace $H^\star_i$ with multiple trees of rank $2$ and at most one of rank $3$, improving by this the value of $F$ (see \autoref{fact:max-cc}).
Moreover, there is at most one tree $H_i$ of rank $3$, for otherwise we could replace two such trees by three trees of rank $2$.

For $i = 1,2,3$, let $h^\star_i$ be the number of trees of rank $i$ in $F^\star$, and let $h_i$ be the number of trees of rank $i$ in $F$.
By \autoref{fact:max-cc}, it is enough to show that $h^\star_1 + h^\star_2 + h^\star_3 = h_1 + h_2 + h_3$.

Call a tree $H^\star_i$ as \emph{special} if it contains a trivial set $T$ and has rank larger than $1$.
Note that there are no special trees of rank $3$, for otherwise we could improve $F^\star$ by breaking such tree into two separate trees of rank 1 and~2.

We claim that there is at most one special tree in $F^\star$.
For otherwise, there are at least two trivial sets $T_i$ and $T_j$ which are spanned by rank-2 trees $H^\star_i$ and $H^\star_j$, respectively.
Then we could replace $H^\star_i,H^\star_j$ by three trees: a spanning tree of $G[T_i]$, a spanning tree of $G[T_j]$, and a spanning tree of all other terminals covered by $H^\star_i$ and $H^\star_j$.
Again, this would contradict the optimality of $F^\star$.

Observe that, if there is a special tree in $F^\star$, then $h^\star_3 = 0$, for otherwise we could improve upon $F^\star$.
Moreover, we can assume that if there is a special tree in $F^\star$, then $h^\star_2 = 1$: for if there was an additional rank-2 tree, we could replace the special tree and construct a trivial tree and a single rank 3 tree.

Finally, by construction of \autoref{greedy-more-1-2} it follows that if there is no special tree in $F^\star$, then $h^\star_1 = h_1$, $h^\star_2 = h_2$ and $h^\star_3 = h_3$.
If there is a unique special tree in $F^\star$, then we have $h^\star_2 = 1$ and $h^\star_3 = 0$.
Therefore, all terminals sets except one are trivial, and \autoref{greedy-more-1-2} finds a solution with $h^\star_1 = h_1$, $h^\star_2 = h_2$ and $h^\star_3 = h_3$.
This completes the proof of \autoref{thm:more-than-1-2-dense}.

\section{Approximation Algorithm for Pairwise $\frac{1}{2}$-Dense Steiner Forest}
\label{sec:1-2-dense}
In this section we give an approximation algorithm for pairwise $\frac{1}{2}$-dense {\sc Steiner Forest} when there are no Steiner vertices.
The algorithm uses, as a subroutine, an approximation algorithm for the {\sc 3-Set Packing} problem.
The currently best known approximation algorithm for {\sc 3-Set Packing} is a $(\frac{4}{3}+\varepsilon)$-approximation, and is due to Cygan~\cite{Cygan2013}.

We will now prove \autoref{thm:1-2-pairwise-dense}.
First note that \autoref{observation:connected-pair} is no longer true for $\delta = \frac{1}{2}$.
But a weakening still holds.
Let $\mathcal T$ be a family of at least two terminal sets, and consider the induced subgraph $G_{\mathcal T} := G[\bigcup_{T \in \mathcal T} T]$.
\begin{observation}
\label{observation:connected-or-2-components}
  The graph $G_{\mathcal T}$ is either connected or consists of two connected components.
\end{observation}
\begin{proof}
  Consider any component $C$ of $G[\bigcup_{T \in \mathcal T} T]$.
  Density implies that $C$ contains at least $\frac{1}{2} |T_i|$ terminals from $T_i$ for each $T_i \in \mathcal T$.
  Thus, $G[\bigcup_{T \in \mathcal T} T]$ has at most two components.
  Observe also, that when $G[\bigcup_{T \in \mathcal T} T]$ has two components $C_1, C_2$, then both $C_1$ and $C_2$ contain exactly half of all terminals from each set $T_i \in \mathcal T$.
  Moreover, terminals from two sets $T_i, T_j \in \mathcal T$ within each component are connected via a biclique.
\qed
\end{proof}

\begin{observation}
\label{observation:contains-trivial}
  If $\mathcal T$ contains a trivial set, then $G_{\mathcal T}$ is connected.
\end{observation}

Let us now define the notion of ``triplets''.
We say that three distinct terminal sets $T_i, T_j, T_\ell$ form a \emph{triplet} if $G[T_i \cup T_j \cup T_\ell]$ is connected.

Now, we claim that triplets provide connectivity for larger structures.
\begin{lemma}
\label{lemma:triplet-inside-connected}
  Let $\mathcal T$ be a family of terminal sets such that $G_{\mathcal T}$ is connected.
  If~$\mathcal T$ contains at least four sets then it contains a triplet.
\end{lemma}
\begin{proof}
  Suppose, for sake of contradiction, that $\mathcal T$ contains no triplets.
  We will show that this implies that $G_{\mathcal T}$ is disconnected, yielding the desired contradiction.

  Let $\mathcal T = \{T_1,\dots,T_{|\mathcal T|}\}$.
  We will show that $G[\bigcup_{i=1}^{z}T_i]$ consists of two connected components, for $z = 2,\dots,|\mathcal T|$.

  Let $z = 2$.
  Observe that there is no pair $T_i, T_j \in \mathcal T$ such that $G[T_i \cup T_j]$ is connected, for otherwise, for any $T_\ell \in \mathcal T$, sets $T_i, T_j, T_\ell$ would form a triplet.

  Now let $z > 2$ and suppose the claim is true for all smaller values of $z$.
  Let $C_1,C_2$ denote the two connected components of $G[\bigcup_{i=1}^{z}T_i]$.
  For $i =1,\dots,z$, let $A_i=V(C_1)\cap T_i$ and $B_i=V(C_2)\cap T_i$.
  Consider now the two components $D_1,D_2$ of $G[T_z \cap T_{z+1}]$.
  Then $A_z = D_1 \cap T_z$ or $A_z = D_2 \cap T_z$; for otherwise, $T_1, T_z, T_{z+1}$ would form a triplet.
  Assume, without loss of generality, that $A_z = D_1 \cap T_z$.
  Let now $A_{z+1} = D_1 \cap T_{z+1}$ and $B_{z+1} = D_2 \cap T_{z+1}$.

  To finish the proof, it suffices to show that nodes from $A_{z+1}$ are not connected to nodes in $B_i$ for $i < z$.
  To this end, note that if $A_{z+1}$ was connected to some node in $B_i$, then $T_i, T_z, T_{z+1}$ would form a triplet.
\qed
\end{proof}

Next, we claim that there is an optimal solution $F^\star$ which has at most one tree $H^\star_i$ with rank larger than 3.
To see this, assume that $F^\star$ contains two trees $H^\star_i,H^\star_j$ of rank larger than 3.
By \autoref{lemma:triplet-inside-connected}, $H^\star_i$ contains a triplet~$\mathcal T'$.
Thus, we can replace $H^\star_i,H^\star_j$ by a spanning tree of triplet $\mathcal T'$, and a tree spanning the terminals in $H^\star_j$ and remaining terminals of $H^\star_i$.
This change does not increase the cost of the solution.

We are ready to describe our approximation algorithm for {\sc Steiner Forest} on pairwise $\frac{1}{2}$-dense instances.
Consider an instance $(G;T_1,\dots,T_k)$.

As a first step, the algorithm deletes all trivial terminal sets from the instance except one.
This step is justified as follows.
Recall that a tree~$H_i$ of a solution is special if it contains a trivial set $T_i$ and has rank larger than~$1$.
We claim that there exists optimal solution $F^\star$ which has at most one special tree.
Too see this, assume that $H^\star_i,H^\star_j$ are special trees in $F^\star$, and $T_i,T_j$ are corresponding trivial sets inside $H^\star_i$ and $H^\star_j$, respectively.
Now, we can replace $H^\star_i,H^\star_j$ by a spanning tree of $T_i$ as a single tree, and a tree spanning~$T_j$ and all terminals outside $T_i\cup T_j$ which are spanned by $H^\star_i,H^\star_j$.
Here we used \autoref{observation:contains-trivial}.
It is also easy to see that it is arbitrary which trivial set $T$ is used to construct a special tree in $F^\star$, as by \autoref{observation:contains-trivial} its choice does not affect the cost of the solution.


After this preprocessing, the algorithm reduces the problem to an instance of {\sc 3-Set-Packing}.
The full algorithm is described in \autoref{algorithm:apx-1-2-dense}.
\begin{algorithm}[h!]
  \caption{Approximation algorithm for pairwise $\frac{1}{2}$-dense {\sc Steiner Forest}\label{algorithm:apx-1-2-dense}}
  \begin{algorithmic}[1]
    \Require A graph $G$ and terminal sets $T_1,\dots,T_k$ such that for each $t\in T_i$ and $j\in[k]$, $|N_{T_j}(t)|\geq |T_j|/2$.
    \Ensure A $(\frac{13}{12}+\varepsilon)$-approximate Steiner forest $F$ for $(G;T_1,\dots,T_k)$.
    \item Set $F=\emptyset$, $\varepsilon > 0$, and let ${\cal T}_{triv}$ be the set of all trivial terminal sets.
    \item Construct an instance $I_{\textsc{pack}} = (U,\mathcal F)$ of {\sc 3-Set Packing}: let $\mathcal{U} = \{T_1,\dots,T_k\} \setminus {\cal T}_{triv}$, and let $\mathcal F$ be as follows:
      \begin{enumerate}
        \item Add to $\mathcal F$ all pairs $(T_i, T_j)$ of terminal sets in $\mathcal{U}$ for which $G[T_i \cup T_j]$ is connected.
        \item Add to $\mathcal F$ all triplets from $\mathcal{U}$.
      \end{enumerate}
    \item Obtain a $\left(\frac{4}{3-\varepsilon}\right)$-approximate solution $\mathcal F_{\textsc{pack}}\subseteq \mathcal F$ to $I_{\textsc{pack}}$. 
    \item If ${\mathcal F }_{\textsc{pack}}= \emptyset = {\cal T}_{triv}$, and $\mathcal{U} \neq \emptyset$, then output ``No solution" and exit.
    \item Add to $F$ spanning trees of every trivial set in ${\cal T}_{triv}$.
    \item Add to $F$ spanning trees of the sets in $\mathcal F_{\textsc{pack}}$.
    \item If some terminals remain uncovered, then pick any tree in $F$ and extend it to connect these remaining terminals into a single tree.
    \item Return $F$.
  \end{algorithmic}
\end{algorithm}

It remains to analyze the approximation ratio of \autoref{algorithm:apx-1-2-dense}.
Let $t = |{\cal T}_{triv}|$, let $\textrm{opt}_{\textsc{pack}}$ be the value of an optimal solution to $I_{\textsc{pack}}$, and let $\textrm{sol}_{\textsc{pack}}$ be the value of the solution to $I_{\textsc{pack}}$ computed by the $(4/3 + \varepsilon)$-approximation algorithm for {\sc 3-Set Packing}, due to Cygan~\cite{Cygan2013}.
By \autoref{fact:max-cc} and above observations about the structure of optimal solutions, it remains to find an $\alpha$ such that
\begin{equation*}
  |V(G)| - (\textrm{sol}_{\textsc{pack}} + t) \leq \alpha \cdot (|V(G)| - (\textrm{opt}_{\textsc{pack}} + t)) \enspace .
\end{equation*}
Observe now, that $|V(G)| -t \geq 4 \cdot \textrm{opt}_{\textsc{pack}}$, as we removed all trivial sets, and so every terminal set consists of at least two vertices.
Thus, we have that
\begin{align*}
  4\cdot \textrm{opt}_{\textsc{pack}} & \leq |V(G)| - t \\
  \Rightarrow(4+4\varepsilon)\cdot \textrm{opt}_{\textsc{pack}} &\leq (1+\varepsilon)\cdot (|V(G)| - t) \\
    \Rightarrow\left(\frac{13+\varepsilon}{12} - \frac{9-3\varepsilon}{12}\right) \cdot \textrm{opt}_{\textsc{pack}} &\leq \left(\frac{13+ \varepsilon}{12} - 1\right)\cdot (|V(G)| - t) \\
    \Rightarrow(|V(G)| - t) - \frac{3-\varepsilon}{4} \cdot \textrm{opt}_{\textsc{pack}} &\leq \frac{13+\varepsilon}{12}\left(|V(G)| - t - \textrm{opt}_{\textsc{pack}}\right)\\
    \Rightarrow|V(G)| - (\textrm{sol}_{\textsc{pack}}+t) &\leq \frac{13+\varepsilon}{12} \cdot \left(|V(G)| - (\textrm{opt}_{\textsc{pack}}+t)\right),
\end{align*}
where in the last line we used that $\textrm{sol}_{\textsc{pack}} \geq \frac{3-\varepsilon}{4} \cdot \textrm{opt}_{\textsc{pack}}$.
This shows that we can set $\alpha=\frac{13+\varepsilon}{12}$ for any $\varepsilon > 0$.
This completes the proof of \autoref{thm:1-2-pairwise-dense}.

\section{Hardness of Approximation for Dense Group Steiner Tree}
\label{sec:group-st}
In this section, we show that the {\sc Dense Group Steiner Tree} is $\mathsf{APX}$-hard. 
Our reduction takes an instance $(U,\mathcal S)$ of {\sc Set Cover}, where $U$ is an $n$-element set and $\mathcal S$ is a family of non-empty subsets of $U$, and constructs an instance of {\sc Group Steiner Tree} as follows.
For each element $u\in U$ create a terminal group $T_u = \{e_{u,1}, e_{u,2}, e_{u,k}\}$ where $k = |\mathcal S|$.
For each set $S_j\in\mathcal S$ create a node $s_j$ and connect $s_j$ to all nodes $e_{u,j}$ for which $u \in S_j$.
Put all the $s_1, \dots, s_k$ into a single group $T_0$.
Finally, add a root node $r$ forming a single group $T_r$ and connect all nodes $s_j$ to $r$.
See \autoref{fig:dense-group-hardness} for an illustration.
\begin{figure}[ht]
  \centering
  \includegraphics[width=0.7\textwidth]{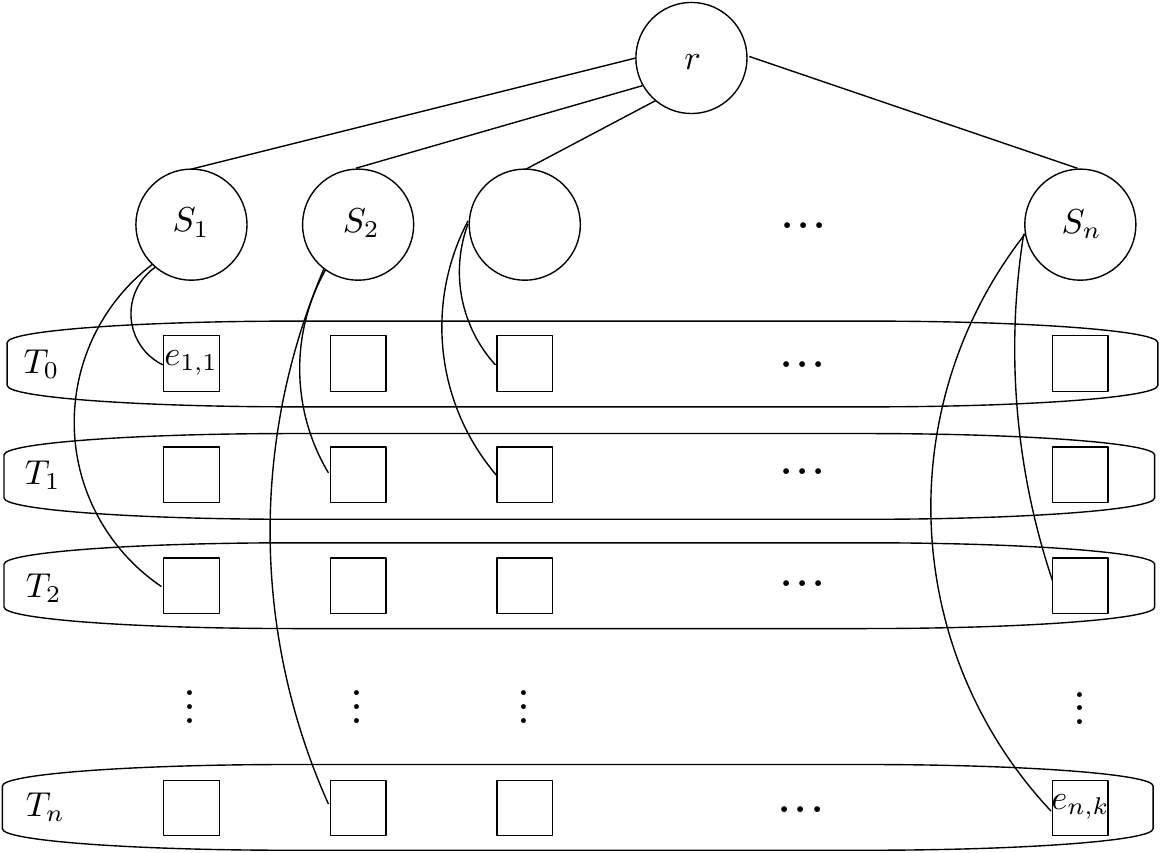}
  \caption{Reducing {\sc Set Cover} to {\sc Group Steiner Tree} on dense instances.\label{fig:dense-group-hardness}}
\end{figure}

Observe, that each node belongs to the exactly one group, so there are no Steiner vertices.
Therefore, the density condition holds trivially.
Now any feasible solution to instance $(U,\mathcal S)$ of cost $c$ corresponds to a feasible solution to instance $(G;T_0,T_1,\dots,T_{|U|},T_r)$ of cost $n+c$, and vice-versa.
As {\sc Set Cover} is hard to approximate with factor $\Theta(\log n)$, as shown by Feige~\cite{Feige1998}, and any optimal solution to $(U,\mathcal F)$ can be assumed have value $\Theta(n)$, the $\mathsf{APX}$-hardness of {\sc Group Steiner Tree} for dense instances follows.

\medskip
\noindent
\textbf{Acknowledgements.} M.L. and S.M.M. were supported by NCN grant number 2015/18/E/ST6/00456.
M.M. was supported by DFG grant MN~59/4-1.
We thank Jaros{\l}aw Byrka for helpful discussions.

\bibliographystyle{abbrv}
\bibliography{references}

\begin{thebibliography}{1}

\bibitem{AgrawalKR1995}
A.~Agrawal, P.~N. Klein, and R.~Ravi.
\newblock When trees collide: An approximation algorithm for the generalized
  {S}teiner problem on networks.
\newblock {\em {SIAM} J. Comput.}, 24(3):440--456, 1995.

\bibitem{ByrkaGRS2013}
J.~Byrka, F.~Grandoni, T.~Rothvo{\ss}, and L.~Sanit{\`{a}}.
\newblock Steiner tree approximation via iterative randomized rounding.
\newblock {\em J. {ACM}}, 60(1):6:1--6:33, 2013.

\bibitem{Cygan2013}
M.~Cygan.
\newblock Improved approximation for 3-dimensional matching via bounded
  pathwidth local search.
\newblock In {\em Proc. FOCS 2013}, pages 509--518, 2013.

\bibitem{DinurS2004}
I.~Dinur and S.~Safra.
\newblock On the hardness of approximating minimum vertex cover.
\newblock {\em Ann. Math.}, 162:2005, 2004.

\bibitem{Feige1998}
U.~Feige.
\newblock A threshold of $\ln(n)$ for approximating set cover.
\newblock {\em J. {ACM}}, 45(4):634--652, 1998.

\bibitem{Hauptmann2013}
M.~Hauptmann.
\newblock On the approximability of dense {S}teiner problems.
\newblock {\em J. Discrete Algorithms}, 21:41--51, 2013.

\bibitem{KarpinskiZ1997}
M.~Karpinski and A.~Zelikovsky.
\newblock New approximation algorithms for the {S}teiner tree problems.
\newblock {\em J. Comb. Optim.}, 1(1):47--65, 1997.

\bibitem{WilliamsonGMV1995}
D.~P. Williamson, M.~X. Goemans, M.~Mihail, and V.~V. Vazirani.
\newblock A primal-dual approximation algorithm for generalized {S}teiner
  network problems.
\newblock {\em Combinatorica}, 15(3):435--454, 1995.

\end{thebibliography}

\appendix
\section{Proof of Theorem 1}
\label{sec:proof-unit-weight-thm}
\begin{proof}
  We give a reduction from {\sc  Vertex Cover} in graphs of bounded maximum degree.
  The reduction has two steps: first we will reduce to unit-weight {\sc Steiner Tree} (with Steiner vertices), and second replace Steiner nodes by pairs of terminals, obtaining the desired {\sc Steiner Forest} instance.

  Let $G$ be a graph constituting an instance of {\sc Vertex Cover}.
  Let $G'$ be the graph whose node set $V(G')$ is a union of $V(G)$ and $E(G)$, and designate all nodes in $G'$ corresponding to edges of $G$ as terminals: so $T = E(G)$.
  Now, connect each terminal in $G'$  to two Steiner nodes corresponding to the endpoints of an edge in~$G$.
  Finally, connect all the Steiner vertices $V(G)$ with each other, so that they form a clique.

  We claim that each vertex cover $C$ in $G$ corresponds to a Steiner tree of size $E(G) + |C| - 1$ for $(G',T)$.
  This can be seen by taking an edge for each terminal to some node corresponding to a node of $C$, and then adding a spanning tree of nodes corresponding to a vertex cover.
  As {\sc Vertex Cover} is $\mathsf{APX}$-hard in graphs of bounded maximum degree~\cite{DinurS2004} and $k = \Theta(|V(G)|) = \Theta(|E(G)|)$, the $\mathsf{APX}$-hardness for {\sc Steiner Tree} follows.

  For the second part of the reduction, for each node $v\in E(G)\subseteq V(G')$ add a node $v'$ which is connected by a single edge to $v$.
  Add also set $T_v = (v,v')$ to a family of terminal sets to connect in the {\sc Steiner Forest} instance.
  Finally, make $V(G)$ also a set of terminals to connect.
  The complete instance is depicted in \autoref{fig:reduction-sf-no-steiner}.

  To finish the proof, observe  that solution size increases by $|V(G)|$, which is of the same order as the initial solution.

\begin{figure}[b]
    \centering
    \includegraphics[width=0.7\textwidth]{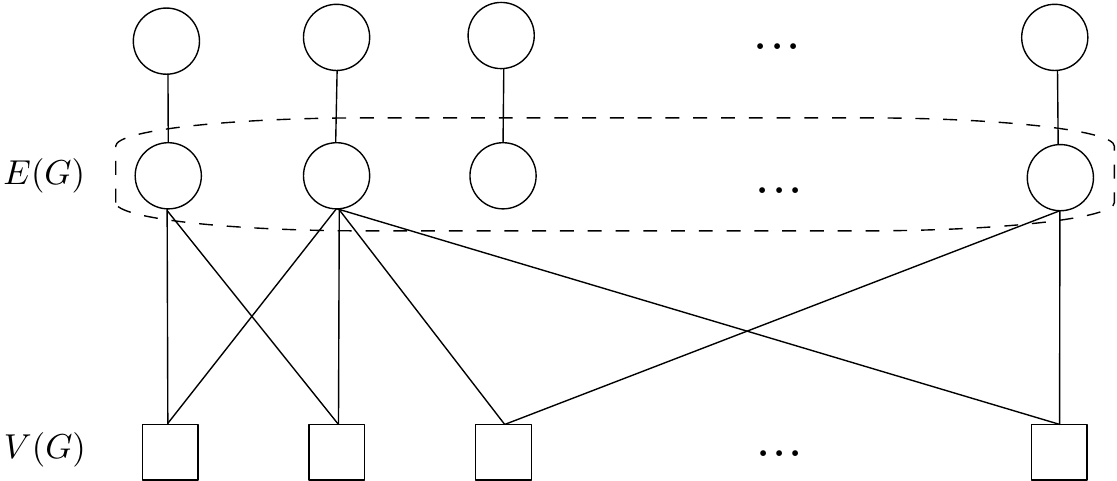}
    \caption{Reduction from {\sc Vertex Cover} in bounded-degree graphs to unit-weight {\sc Steiner Forest} without Steiner nodes.
    The dotted line indicates that nodes of $E(G)$ are connected in a clique.
    The family of terminal sets to connect consists of a set of squares and $|E(G)|$ pairs of circles.\label{fig:reduction-sf-no-steiner}}
  \end{figure}
\qed
\end{proof}

\end{document}